\def\year{2021}\relax
\documentclass[letterpaper]{article} 
\pdfoutput=1 
\usepackage[switch]{lineno}
\usepackage{aaai21}  
\usepackage{times}  
\usepackage{helvet} 
\usepackage{courier}  
\usepackage[hyphens]{url}  
\usepackage{graphicx} 
\usepackage{natbib}  
\usepackage{caption} 
\usepackage[utf8]{inputenc} 
\usepackage{url}            
\usepackage{booktabs}       
\usepackage{amsfonts}       
\usepackage{nicefrac}       
\usepackage{microtype}      
\newsavebox{\measurebox}

\usepackage{bm}
\usepackage{soul}
\usepackage{amsmath}
\usepackage{amsthm}
\usepackage{bbding} 
\usepackage{xcolor}
\usepackage{enumitem}
\usepackage{subcaption}
\usepackage{amsfonts}
\usepackage{algpseudocode}

\usepackage{xspace}

\usepackage{amssymb}

\usepackage{algorithm}

\usepackage{mathrsfs}

\usepackage{etoolbox}
\usepackage{diagbox}
\usepackage{multirow}
\usepackage{rotating} 
\usepackage{lipsum}

\graphicspath{{./FIG/}}

\newtheorem{problem}{Problem}

\newtheorem{thm}{Theorem}

\newtheorem{example}{Example}

\newcommand{\bit}{\begin{compactitem}}
	\newcommand{\eit}{\end{compactitem}}
\newcommand{\ben}{\begin{compactenum}}
	\newcommand{\een}{\end{compactenum}}

\newcommand{\BComment}[1]{\Comment{\emph{#1}}}
\newcommand{\RNum}[1]{\uppercase\expandafter{\romannumeral #1\relax}}

\newcommand{\atn}[1]{\textcolor{red}{#1}}

\renewcommand{\atn}[1]{\textcolor{black}{#1}}

\newcommand{\zhangatn}[1]{\textcolor{black}{#1}}

\newcommand{\method}{\textsc{AugSplicing}\xspace}
\newcommand{\algom}{our algorithm\xspace}
\newcommand{\tensor}[1]{\mathcal{#1}}   

\newcommand{\TX}{\tensor{X}}

\newcommand{\TB}{\tensor{B}}

\newcommand{\spot}{\textsc{Spotlight}\xspace}

\newcommand{\dcube}{\textsc{D-Cube}\xspace}
\newcommand{\mzoom}{\textsc{M-Zoom}\xspace}
\newcommand{\sambaten}{{SamBaTen}\xspace}
\newcommand{\densestream}{\textsc{DenseStream}\xspace}

\newcommand{\cross}{\textsc{CrossSpot}\xspace}
\newcommand{\eigenspokes}{{EigenSpokes}\xspace}

\newcommand{\cpd}{{CPD}\xspace}
\newcommand{\catchcore}{{CatchCore}\xspace}

\usepackage{tikz}
\newcommand*{\circled}[1]{\lower.7ex\hbox{\tikz\draw (0pt, 0pt)%
    circle (.5em) node {\makebox[1em][c]{\small #1}};}}

\title{\method: Synchronized Behavior Detection in Streaming Tensors}
\author{
Jiabao Zhang\textsuperscript{\rm 1}, Shenghua Liu\textsuperscript{\rm 1 \#}, Wenting Hou\textsuperscript{\rm 2}, Siddharth Bhatia\textsuperscript{\rm 3}, \\
Huawei Shen\textsuperscript{\rm 1}, Wenjian Yu\textsuperscript{\rm 4 \#}, Xueqi Cheng\textsuperscript{\rm 1}   \\    
}

\affiliations{
\textsuperscript{\rm 1}	Institute of Computing Technology, Chinese Academy of Sciences\\
\textsuperscript{\rm 2} Beijing InnovSharing Co.Ltd\\
\textsuperscript{\rm 3} National University of Singapore\\
\textsuperscript{\rm 4}Dept. Computer Science \& Tech., Tsinghua University\\
{\{zhangjiabao,liushenghua,shenhuawei,cxq\}@ict.ac.cn, siddharth@comp.nus.edu.sg, yu-wj@tsinghua.edu.cn}
}

\begin{document}
\maketitle

\begin{abstract}
How can we track synchronized behavior in a stream of time-stamped tuples, such as mobile devices installing and uninstalling applications in the lockstep, to boost their ranks in the app store?  We model such tuples as entries in a streaming tensor, which augments attribute sizes in its modes over time. Synchronized behavior tends to form dense blocks (i.e.~subtensors) in such a tensor, signaling anomalous behavior, or interesting communities. However, existing dense block detection methods are either based on a static tensor, or lack an efficient algorithm in a streaming setting. 
Therefore, we propose a fast streaming algorithm, \method, which can detect the top dense blocks by incrementally splicing the previous detection with the incoming ones in new tuples, avoiding re-runs over all the history data at every tracking time step.
\method is based on a splicing condition that guides the algorithm (Section \ref{sec:method}). Compared to the state-of-the-art methods, our method is (1) effective to detect fraudulent behavior in installing data of real-world apps and find a synchronized group of students with interesting features in campus Wi-Fi data; (2) robust with splicing theory for dense block detection; (3) streaming and faster than the existing streaming algorithm, with closely comparable accuracy.
    \label{sec:abs}
\end{abstract}

\section{Introduction}
\label{sec:intro}

Given a stream of time-stamped tuples ($a_1$, $a_2$, $\cdots$, $a_n$, $t$), how can we spot the most synchronized behavior up to now in real-time?
{\let\thefootnote\relax\footnote { \# Corresponding authors     \label{ Corresponding-author }     }} \par 
\setcounter{footnote}{0} Such a problem has many real applications. In online review sites such as Yelp, let $a_1$ be a user, $a_2$ be a restaurant, $a_3$ be a rating score, and $t$ be the rating time. The most synchronized rating behavior of high scores indicates the most suspicious review fraud~\cite{hooi2016fraudar,jiang2014catchsync}. In application logs, $a_1$, $a_2$, $a_3$, and $t$
can represent a mobile device,
an app, installing time, and uninstalling time respectively. 
Highly synchronous installation and uninstallation from a
group of devices can reveal the most suspicious behavior of boosting target apps' ranks in an app store.
In terms of pattern discovery, synchronous connections and disconnections to the Wi-Fi access point (AP) in campus Wi-Fi connection logs can discover students that have the same classes of interest. 

Let such time-stamped tuples be entries of a tensor with multiple dimensions, such as \emph{user}, \emph{object}, and \emph{time} (Figure~\ref{fig:tensor_example}). 
Note that we call each dimension as a mode like \cite{lu2013multilinear}, and a two-mode tensor is a matrix.
Since tensors allow us to consider additional information especially the time,
the densest block (subtensor) of interest can identify the most synchronized behavior in time-stamped tuples~\cite{jiang2015general,shah2015timecrunch,shin2017d}.

In such a streaming tensor, the attribute size of \emph{time} mode is augmented over
time as shown in Figure~\ref{fig:tensor_example}. Other modes such as \emph{user}
and \emph{object} can also be augmented when an unseen user or object is observed.
    \zhangatn{Nowadays, dense subtensor detection methods for streaming tensors are essential. This is because it is much easier than in the past to collect large datasets with the advance of technology. Not only is the size of real data very large, but also the rate at which it arrives is high~\cite{akoglu2015graph}. For example, Facebook users generate billions of posts every day, billions of credit card transactions are performed each day, and so on. As such, whole data may be too large to fit in memory or even on a disk. On the other hand, we can think of this kind of data generation as streaming tensors as mentioned above. Thus, the methods which can update their estimations efficiently when the tensor changes over time are essential for dense subtensor detection problem.} 
However, many existing works on dense subtensor detection were
designed for static tensors given in a batch~\cite{shin2016mzoom,shin2017d,yikun2019no} and we refer to them as batch algorithms.
Although these batch algorithms are near-linear with the size of tuples (i.e. non-zero
entries) in a tensor,
re-running the algorithms at every time step for a streaming tensor can result in memory overload when we meet huge size datasets and quadratic time complexity. This causes limited scalability in a streaming setting due to the repeated computation on past tuples \cite{teng2016scalable}.
As for the state-of-the-art streaming algorithm,
\densestream~\cite{shin2017densealert}, maintained a fine-grained order (i.e.~D-order) to search for the densest subtensor. The order is updated for every single new tuple, limiting the detection speed.

Therefore we propose \method, a fast and incremental algorithm to approximate the up-to-date dense blocks in streaming tensors.
Without re-running batch algorithms, our heuristic algorithm based on the splicing condition reduces the search space, incrementally splices dense blocks of previous detections and
the new blocks detected only in an incoming tensor (right-side tensor in Figure~\ref{fig:tensor_example}). As such, \method can detect dense subtensors at every time step in real-time.
Experiments show that \method is the fastest, with comparable accuracy with the state-of-the-art methods. In summary, our main contributions are:

\begin{figure*}[t]
\begin{center}
\begin{subfigure}[t]{0.5\textwidth}
\centering
{\includegraphics[width=\linewidth]{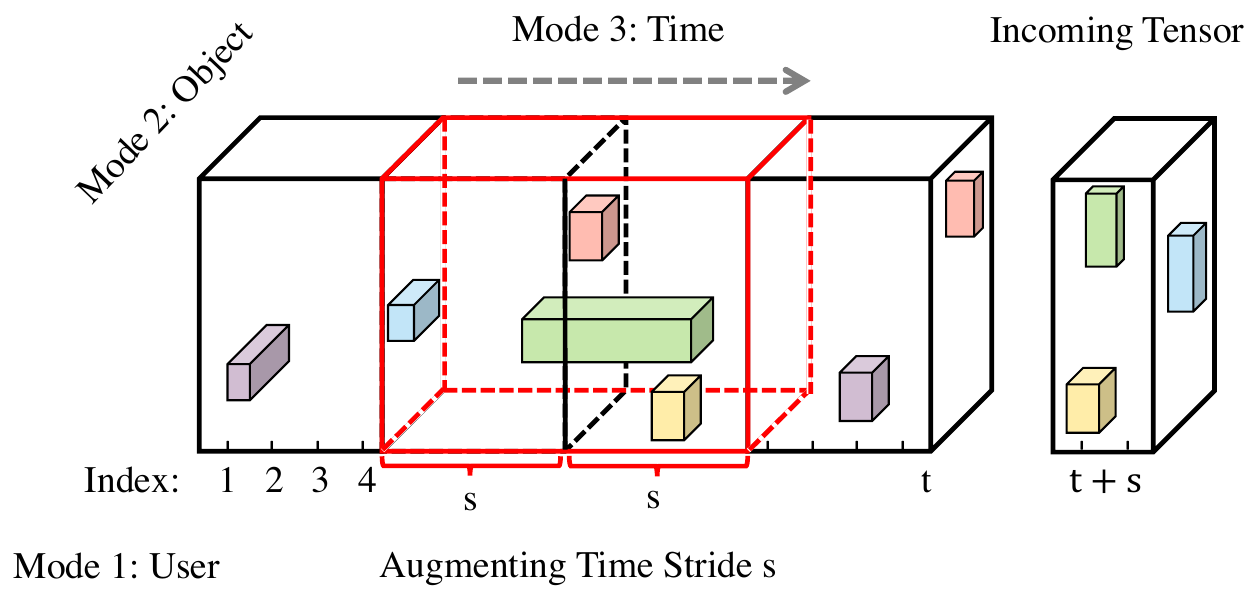}}
\caption{Streaming Tensor}
\label{fig:tensor_example}
\end{subfigure}\hspace{12mm}
\begin{subfigure}[t]{0.36\textwidth}
\centering
{\includegraphics[width=\linewidth]{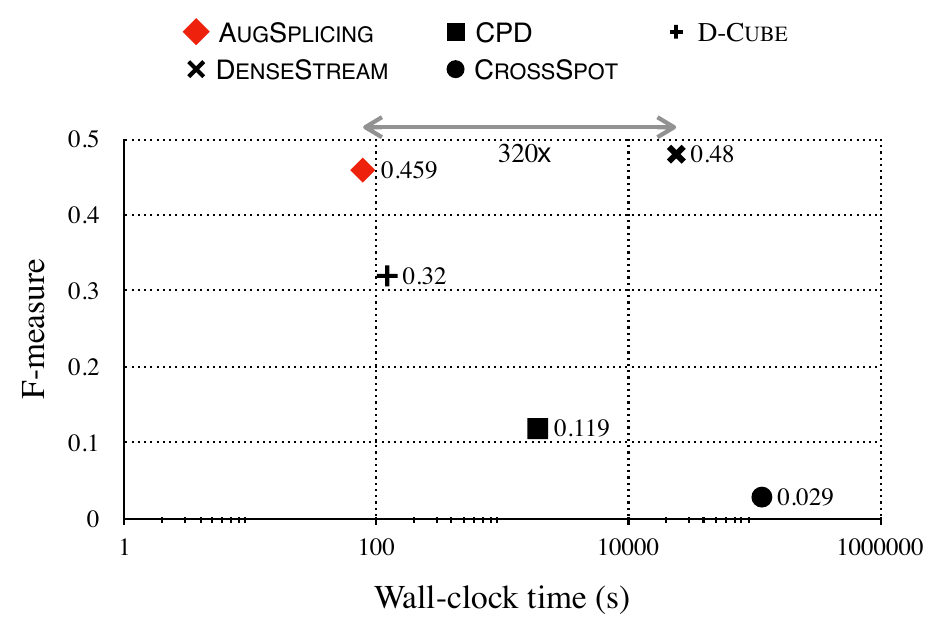}}
\caption{fast and accurate}
\label{fig:realaccurate}
\end{subfigure}
\end{center}
\caption{
    (a) A tensor contains dense blocks (subtensors),
    and an incoming tensor at a time step contains tuples in range ($t, t+s$]. 
    (b)\method is fastest while maintaining similar accuracy (in terms of F-measure) compared to state-of-the-art approaches}
\end{figure*}

\begin{enumerate}
\item \textbf{Fast and Streaming Algorithm:} 
    We propose a fast dense block detection algorithm 
	in streaming tensors, which is up to $320$ times faster than the current state-of-the-art algorithms ( Figure~\ref{fig:realaccurate}).
\item \textbf{Robustness:}
    \method is robust with splicing theory to do incremental splices for dense block detection.
\item \textbf{Effectiveness and Explainable Detection:} 
    Our algorithm achieves accuracy (in terms of F-measure) comparable to the best baseline, \densestream (Figure~\ref{fig:realaccurate}). 
    \method spots suspicious mobile devices that boost target apps' ranks in a recommendation list by synchronous installations and uninstallations in real-world data. 
	The result shows that the suspicious installations of $21$ apps on $686$
	devices mostly happened on the first $6$ days (Figure~\ref{fig:appcase3}) and
	the target apps were uninstalled within $3$ days (Figure~\ref{fig:appcase4}), which is a very unusual synchronized
	behavior among a group of devices and apps.
	Moreover, in real Wi-Fi data, we find a group of students with a similar schedule,
	showing periodic and reasonable activities on the campus (Figure~\ref{fig:wifi_distribution}).
\end{enumerate} 

{\bf Reproducibility}: Our code and datasets are publicly available at \url{https://github.com/BGT-M/AugSplicing}.

\section{Related Work}
\label{sec:related}
Multi-aspect tuples can also be represented as attributed edges in a rich graph, e.g.~users and objects as graph nodes, and rating scores and times as different attributes on graph edges. We, therefore, summarise the related research on 
dense block detection using both graphs and tensors (including two-mode matrices).

\subsection{Static Tensors and Graphs}
Dense subgraph detection has been extensively 
studied in~\cite{hooi2016fraudar,gibson2005discovering,charikar2000greedy}.
Spectral decomposition based methods, e.g., \textsc{SpokEn}~\cite{prakash2010eigenspokes} considers the \eigenspokes 
on EE-plot produced by pairs of eigenvectors to detect near-cliques in social network.
FRAUDAR~\cite{hooi2016fraudar} considers both node and edge suspiciousness as a metric to detect frauds (i.e.~dense blocks) and is also resistant to camouflage. \textsc{CrossSpot}~\cite{jiang2015general} proposes an intuitive, principled metric satisfying the axioms that any metric of suspiciousness should obey, and design an algorithm to spot dense blocks sorted by order of importance (``suspiciousness”).
HOSVD, CP Decomposition (\cpd)~\cite{kolda2009tensor} and disk-based algorithm~\cite{oh2017s} spot dense subtensors by Tensor decomposition. \mzoom~\cite{shin2016mzoom} and 
\dcube~\cite{shin2017d} adopt greedy approximation algorithms to detect dense subtensors with quality guarantees.
\catchcore~\cite{feng2019catchcore}  designs a unified metric optimized with gradient-based methods to find hierarchical dense
subtensors. 
\cite{liu2018contrast} optimizes the metric of suspiciousness from topology, rating time, and scores. 
ISG+D-spot~\cite{yikun2019no} constructs information sharing graph and finds dense subgraphs for the hidden-densest block patterns. Flock~\cite{shah2017flock} detects lockstep viewers in a live streaming platform.

However, these methods do not consider any temporal information, or only treat time bins as a static mode.

\subsection{Dynamic Tensors and Graphs}
In terms of dynamic graphs, some methods monitor the evolution of the entire graph and detect changes (density or structure) of subgraphs. \spot~\cite{eswaran2018spotlight} utilizes graph sketches to detect the sudden density changes of a graph snapshot in a time period.
SDRegion~\cite{wong2018sdregion} detects blocks consistently becoming dense or sparse in genetic networks. EigenPulse~\cite{zhang2019eigenpulse} is based on a fast spectral decomposition approach, single-pass PCA~\cite{yu2017single}, to detect the density surges. 
Other methods, like
\textsc{Midas} \cite{bhatia2020midas,bhatia2020real} and \textsc{MStream}~\cite{bhatia2020mstream} detect suddenly arriving groups of suspiciously similar edges in edge streams, but do not take into account the topology of the graph. \densestream~\cite{shin2017densealert} maintains dense blocks incrementally for
every coming tuple and updates dense subtensors when it meets an updating condition, limiting the detection speed.

As for clustering-based methods, \cite{manzoor2016fast} compare graphs based on the relative frequency of local substructures to spot anomalies. 
\cite{cao2014uncovering} uncovers malicious accounts that act similarly in a sustained period of time. Tensor decomposition-based methods, e.g., \sambaten~\cite{gujral2018sambaten} and OnlineCP~\cite{zhou2016accelerating} conduct the incremental tensor decomposition.
\atn{Summarization based methods, e.g.~\cite{shah2015timecrunch} finds temporal patterns by summarizing important temporal structures. \cite{araujo2014com2} uses iterated rank-1 tensor decomposition, coupled with MDL (Minimum Description Length) to discover temporal communities.}

Our method formulates the time-stamped tuples as a streaming tensor whose \emph{time} mode is constantly augmented,
such that the numerical value of entries in the previously observed tensor will not be changed. We incrementally splice incoming dense subtensors with the previous ones at each time step, achieving efficient results.

\label{sec:definition}
\section{Definitions and Problem}
We now give the notations used in the paper and describe our problem. Table~\ref{tab:symbols} lists the key symbols.
\begin{table}[ht]
\begin{center}

\begin{tabular}{p{0.20\columnwidth}<{\centering}|p{0.70\columnwidth}}
\toprule
{\bf Symbol}  & {\bf Definition} \\
\midrule
 $\TX, \TB$ & a tensor, and subtensor, i.e.~block \\
        $\TX(t)$ & $N$-mode tensor up to time $t$\\
        $N$ & number of modes in $\TX$ \\
        $e_{i_1,\cdots,i_N}(t)$ & {entry of $\TX(t)$ with index $i_1,\cdots,i_N$}\\ 
        $I_n(\cdot)$ & {set of mode-$n$ indices of tensor} \\
        $M(\cdot)$ & mass of tensor i.e.~sum of non-zero entries\\ 
        $S(\cdot)$ & size of tensor\\ 
        $g(\cdot)$ &  arithmetic degree density of tensor \\
        $s$ &  augmenting time stride\\ 
        $\TX(t, s)$ & $N$-mode augmenting tensor within \\ & time range $(t, t+s] $\\
        $k$ & number of blocks kept during iterations\\
        $[x]$ & $\{1,2,\cdots,x\}$\\
\bottomrule
\end{tabular}
\caption{Table of Symbols.}
\label{tab:symbols}
\end{center}
\end{table}

\textbf{Tensors} are multi-dimensional arrays as the high-order generalization of vectors ($1$-dimensional tensors) and matrices ($2$-dimensional tensors). The number of dimensions of a tensor is its order, denoted by $N$.
And each dimension is called a mode.
For an $N$-mode tensor $\TX$ with non-negative entries,
each $(i_1, \cdots, i_N)$-th entry is denoted by $e_{i_1 \ldots i_N}$. We use mode-$n$ to indicate the $n$-th mode as~\cite{de2000best,lu2013multilinear} do. 
Let $i_n$ be mode-$n$ index of entry $e_{i_1 \ldots i_n \ldots i_N}$. We define the mass of $\TX$ as $M(\TX)$ to be the sum of its non-zero entries,
and the size of $\TX$ as $S(\TX)=\sum_{i=1}^{N} |I_i(\TX)|$, where $I_i(\TX)$
is the set of mode-$i$ indices of $\TX$. Let \textbf{block} $\TB$ be a subtensor of $\TX$.
Similarly, $M(\TB)$ and $S(\TB)$ are mass and size of block $\TB$. 

Our problem can be described as follows:
\begin{problem}[\textbf{Synchronized Behavior Detection in Streaming Tensor}]
	Given a stream of time-stamped tuples, i.e.~streaming tensor $\TX(t)$, and an augmenting time stride $s$, find top $k$ dense blocks (i.e. subtensors) of $\TX(t)$ so far at every tracking time step.
\end{problem}

We use the \emph{arithmetic average mass} as the density measure of a block $\TB$ 
to avoid trivial solution as~\cite{shin2016mzoom,shin2017d}, i.e. $g(\TB) = \frac{M(\TB)}{S(\TB)}$.

\label{sec:model}
\section{Proposed Algorithm}
\label{sec:method}

In this section, we first theoretically analyze the splicing condition to increase density, and then guided by theory design a near-greedy algorithm to splice any two dense blocks. The overall algorithm (\method) and time complexity are
given in Section~\ref{secoverallalg}.

\subsection{Theoretical Analysis for Splicing}
\label{secsplicecond}
We analyze the theoretical condition that whether splicing (i.e.~merging partially) two dense blocks
can result in a block with higher density, as Figure~\ref{figspliceall} shows. 
We call such merging as \textbf{splicing}. 

\begin{thm}[Splicing Condition]
Given two blocks $\TB_1$, $\TB_2$ with $g(\TB_1) \geq g(\TB_2)$, 
$\exists \mathcal E \subseteq \TB_2$ such that 
    $g(\TB_1 \cup \mathcal E) > g(\TB_1)$
    if and only if
    \begin{equation}
	M(\mathcal E) > \sum_{n=1}^{N} r_n\cdot g(\TB_1) = Q\cdot g(\TB_1),
	\label{eqspcond}
    \end{equation}
    where $r_n = |I_{n}(\mathcal E) \setminus I_{n}(\TB_1)|$, i.e~the number of new mode-$n$ indices that $\mathcal E$ brings into $\TB_1$. $Q = \sum_{n=1}^{N}r_n$, i.e~the total number of new indices that $\mathcal E$ brings into $\TB_1$.
\label{theorem:SplicingCondition}
\end{thm}

\begin{proof}
    
    First, we prove the ``$\Leftarrow$'' condition.
    Based on the definition of $g(\cdot)$, we have

\begin{align*}
    \nonumber
    g(\TB_1\cup\mathcal E) &= \frac{M(\TB_1)+M(\mathcal{E})}{S(\TB_1)+Q} > 
    \frac{M(\TB_1)+Q\cdot g(\TB_1)}{S(\TB_1) + Q} \\
    &= \frac{S(\TB_1)\cdot g(\TB_1)+Q\cdot g(\TB_1)}{S(\TB_1) + Q}=
    g(\TB_1)
\end{align*}

Similarly, we can prove the ``$\Rightarrow$'' condition.
\end{proof}

We can see that while splicing blocks, 
new indices of size $Q$ are brought into
some modes of $\TB_1$, and only merging the block $\mathcal E$ with a large enough mass satisfying inequation~\eqref{eqspcond},
can increase $g(\TB_1)$.
Based on the theory, we design an effective algorithm to splice blocks as shown later.

\subsection{Splicing Two Blocks}
\label{secsplicetwo}

\begin{algorithm*}[htbp]
\caption{Splice two dense blocks}
\label{alg:splicetwoblocks}
\begin{algorithmic}[1]
\Require two dense blocks: $\TB_1$ and $\TB_2$, with $g(\TB_1) \geq g(\TB_2)$.
\Ensure new dense blocks
    \Repeat
        \Statex \textcolor{blue}{/* minimize the size of new indices, $Q$ */}
        \State $\bm q \leftarrow$ get set of modes that have to bring new indices into $\TB_1$ for splicing
        \State $Q \leftarrow |\bm q|$ \BComment{to minimize $Q$, considering
	only one new index from each mode in $\bm q$}
        \State $\bm H \leftarrow$ an empty max heap for blocks and ordered by block mass
	\For{each combination of new indices $(i_{q_1}, \cdots, i_{q_Q})$, $q\in\bm q$}
            \State $\mathcal E \leftarrow $block with entries $\{e_{i_1 \ldots i_{q_1} \ldots i_{q_Q} \ldots i_N} \in \TB_2 ~|~~ \forall n\in [N]\setminus \bm q, i_n \in I_n(\TB_1) \} $
	        \State push $\mathcal E$ into $\bm H$
        \EndFor
        \Statex \textcolor{blue}{/* maximize $M(\mathcal E)$, given $Q$ */ }
        \For{ $\mathcal E \leftarrow \bm H$.top() } 
            \If{$M(\mathcal E) > Q \cdot g(\TB_1)$ } \BComment{inequation~\eqref{eqspcond}}
            \State $\TB_1, \TB_2$ $\leftarrow$ update $\TB_1 \cup \mathcal E, \TB_2 \setminus \mathcal E$
            \State remove $\mathcal{E}$ from $\bm H$, and re-heapify $\bm H$
            \Else
                 \State \textbf{break}
            \EndIf
        \EndFor
    \Until{no updates on $\TB_1$ and $\TB_2$}
\State \textbf{return} new block $\TB_1$ of higher density, and residual dense block $\TB_2$
\end{algorithmic}
\end{algorithm*}

\begin{figure}[!htb]
        \center{\includegraphics[width=0.91\columnwidth]
        {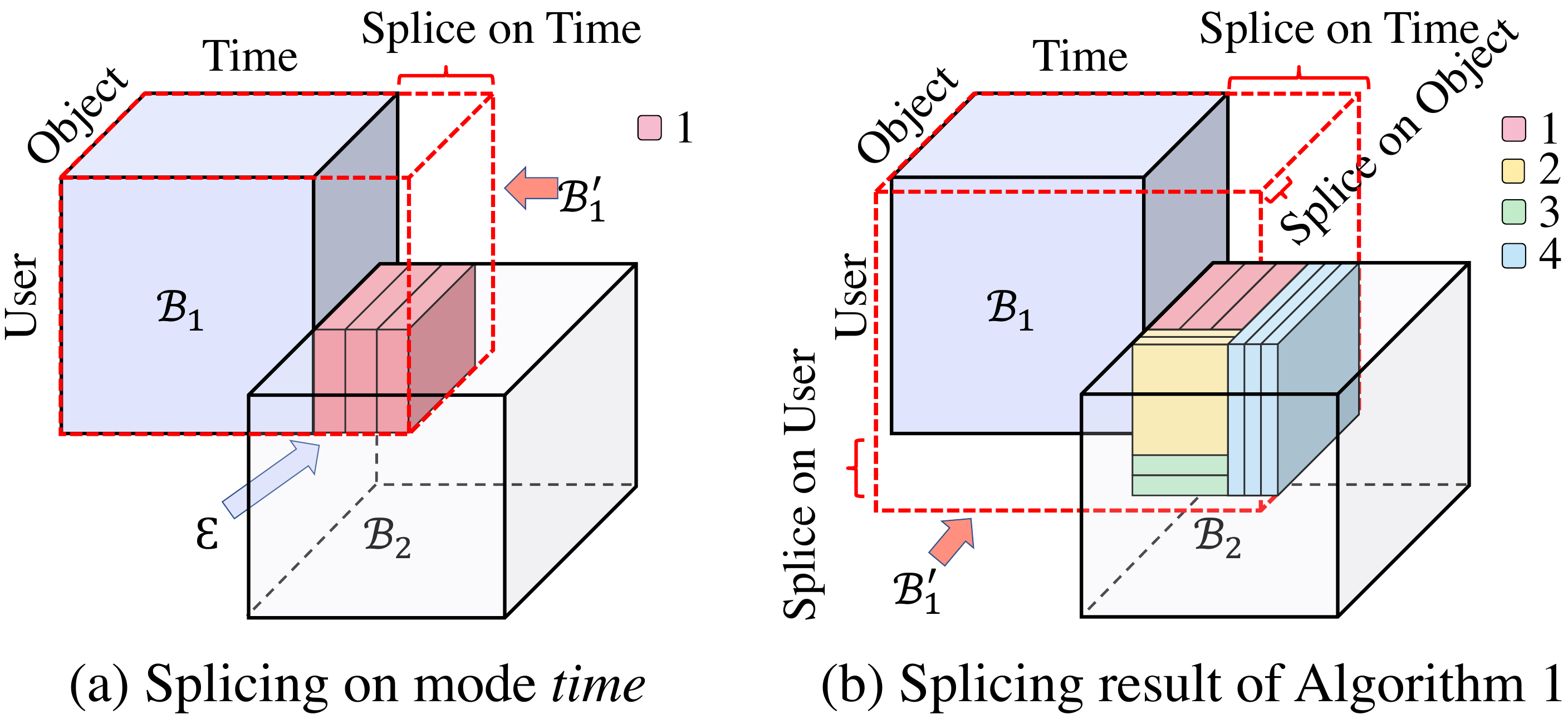}}
        \caption{An illustration of splicing two blocks of three modes, 
    i.e.~($user, object, time$). Given $\TB_1$ and $\TB_2$ with $g(\TB_1) \geq g(\TB_2)$ since there are no common indices on mode $time$, set $\bm q=\{time\}$ and $Q=1$.
	After splicing red blocks \textit{1} into $\TB_1$, all modes of two blocks are overlapped, and Algorithm~\ref{alg:splicetwoblocks} chooses one mode to bring new indices. As shown in (b), colored blocks \textit{2}, \textit{3}, \textit{4} are successively spliced into $\TB_1$, bringing new indices into all three modes of $\TB_1$.
$\TB_1'$ is new $\TB_1$ with higher density.}
 \label{figspliceall}
\end{figure}

The purpose of splicing two blocks is to find a higher-density block by moving
entries from one to another.
Thus, based on the above analysis, the smaller size of new indices (i.e.~smaller $Q$), and a larger mass
of the merging block, can greedily increase the density of the spliced block.
Our algorithm for splicing two given dense blocks is designed as
Algorithm~\ref{alg:splicetwoblocks}. 

Given $g(\TB_1) \geq g(\TB_2)$, the idea is to alternatively find the minimum size of new indices (i.e.~minimum $Q$), and maximum mass of block $\mathcal E$, given new indices for splicing. 

In terms of $Q$, we first need to decide the set of modes $\bm q$ that have to bring new indices into $\TB_1$ (line 2). Since there are no common indices in mode-$q$ of block $\TB_1$ and $\TB_2$, 
at least one new index has to be added to the mode-$q$ indices of $\TB_1$, then we add $q$ to $\bm q$. Thus the minimum $Q=|\bm q|$.
If all modes have common indices, then $\bm q$ is empty and we do
the following:
\begin{enumerate}
    \item Let the block $\mathcal E \subseteq \TB_2$ consist of entries of common indices. 
    \item Move non-zero entries of $\mathcal E$ into $\TB_1$ (if they exist), which increases the density of $\TB_1$ without bringing new indices.
    \item Choose one mode $q$ to splice. For each mode of $[N]$, we generate subblocks of $\TB_2$ by choosing one new index on this mode, and all indices overlapped with $\TB_1$ on other modes.
    Subblock with maximum mass was generated from mode $q$. In such a scenario, $Q=1$ to choose only one mode to splice. 
\end{enumerate}

For mass maximization, we use a max heap to organize blocks by their mass (line 4). Top of the max heap is always the maximum mass block. Then we enumerate all possible combinations of
a new index from each mode in $\bm q$ (lines 5-7) to build a max heap $\bm H$. 
Since the number of modes of blocks, $N$, is usually as small as $3\sim 5$ for real data, 
and the size of possible combinations is comparable to 
$S(\TB_2)$, given $\TB_2$ is a small-size block in original tensor $\TX$.
Moreover, according to inequation~\eqref{eqspcond}, only those blocks with
large enough masses are add into max heap $\bm H$.
Then we splice a maximum-mass block on top of $\bm H$,
iteratively increasing $g(\TB_1)$ 
and merging next top block satisfying $M(\mathcal E) > Q \cdot g(\TB_1)$,
until no large-mass blocks remain for merging (lines 8-13).

With first getting the minimum size of new indices, i.e.~minimum $Q$, and constantly merging maximum-mass block by choosing new indices into $\TB_1$, our algorithm ends until no updates can be made on $\TB_1$ and $\TB_2$.

\begin{example}
Figure~\ref{figspliceall} gives a running example of our algorithm. 
    In the beginning, $\TB_1$ and
$\TB_2$ have no common indices on mode $time$, thus $\bm q=\{3\}$.
Alg~\ref{alg:splicetwoblocks} splices on mode $time$
with red blocks \textbf{1} merged into $\TB_1$, forming new $\TB_1$ of higher
    density (i.e.~$\TB_1'$ in Figure~\ref{figspliceall}(a)). Note that each red block brings only one new index into $\TB_1$, i.e.~$Q=1$. 
    Afterward, all modes of two new blocks have common indices.
    Since $\TB_2$ doesn't have any non-zero entry of common indices with $\TB_1$, Alg~\ref{alg:splicetwoblocks} has to choose one mode $q$ to bring new indices into $\TB_1$ for splicing.
    $q$ is successively the mode $object$, $time$, $user$ in the example. In the end, colored blocks \textbf{1}, \textbf{2}, \textbf{3}, \textbf{4}
    are successively merged. A new block $\TB_1$ (i.e.~$\TB_1'$) with higher density, and residual block $\TB_2$ are returned (Figure~\ref{figspliceall}(b)).
\end{example}

\subsection{Overall Algorithm}
\label{secoverallalg}
In this section, we first describe the overall algorithm for incrementally detecting the densest blocks at each time step $t+s$, then analyze the time complexity of \method, which is near-linear with the number of non-zero tuples. 

Let bold symbols $\bm B(t)$ and $\bm C(t)$ be sets of
top $k+l$ dense blocks of previous $\TX(t)$ and incoming $\TX(t, s)$, where $l$ is a slack constant for approximating the top $k$ dense blocks with $l$ more blocks. 
Our overall algorithm is as follows:

\textbf{(a) Splice two dense blocks:} We iteratively choose two candidate blocks from $\bm B(t) \cup \bm C(t)$, 
denoted as $\TB_1$ and $\TB_2$ with $g(\TB_1) \geq g(\TB_2)$, 
then use Algorithm~\ref{alg:splicetwoblocks} to splice them. 

\textbf{(b) Iteration and Output:}
The splicing iteration stops until no blocks can be spliced, or reach the maximum number of epochs. 
Then, \method outputs top $k$ of $k+l$ dense blocks at time step
$t+s$, and moves on to the next time step with $k+l$ blocks.

\begin{thm}[Time Complexity]
The time complexity of \method at time step $t+s$ is $O(N^{2}(k+l)nnz(\TX(t,s))L(\TX(t,s)) + 2(k+l)nnz(\TB)\log S(\TB))$, where $L(\cdot)=\max_{n \in [N]}{|I_n(\cdot)|}$ and $nnz(\cdot)$ is the number of non-zero entries. 
\end{thm}

\begin{proof}
    At time step $t+s$, 
    the incoming tensor is $\TX(t,s)$, and the complexity for detecting 
    new top $k+l$ dense blocks is $O(N^{2}(k+l)nnz(\TX(t,s))L(\TX(t,s))$
    according to [20].

    Let $\TB$ be a block of the maximum non-zero entries,
    and the largest size among splicing blocks.
    Compared to building a max heap, a more time-consuming procedure is the iteration of updating and re-heapifying when new entries are merged into a block.
    Considering the worst case that all the blocks are merged into one, at most
    $2(k+l)nnz(\TB)$ entries are spliced, i.e.~the maximum number of updates in the max heap. 
    Therefore the time complexity for iterative splices is at most
    $O(2(k+l)nnz(\TB)\log S(\TB)))$, as the heap size is $O(S(\TB))$.
    Thus the complexity of \method at time step $t+s$ is $O(N^{2}(k+l)nnz(\TX(t,s))L(\TX(t,s)) + 2(k+l)nnz(\TB)\log S(\TB))$.
\end{proof}
Since $nnz(\TB) = O( nnz(\TX(t,s)) )$ for proper stride $s$,
our algorithm is near-linear in the number of incremental 
tuples $nnz(\TX(t,s))$ in practice as \cite{shin2017d} shows, which ensures near-linear in the number non-zero entries of streaming tensors. 

\section{Experiments}
\label{sec:experi}
\begin{figure*}[t]
\centering
\begin{subfigure}{0.35\textwidth}
  \centering
  {\includegraphics[width=\linewidth]{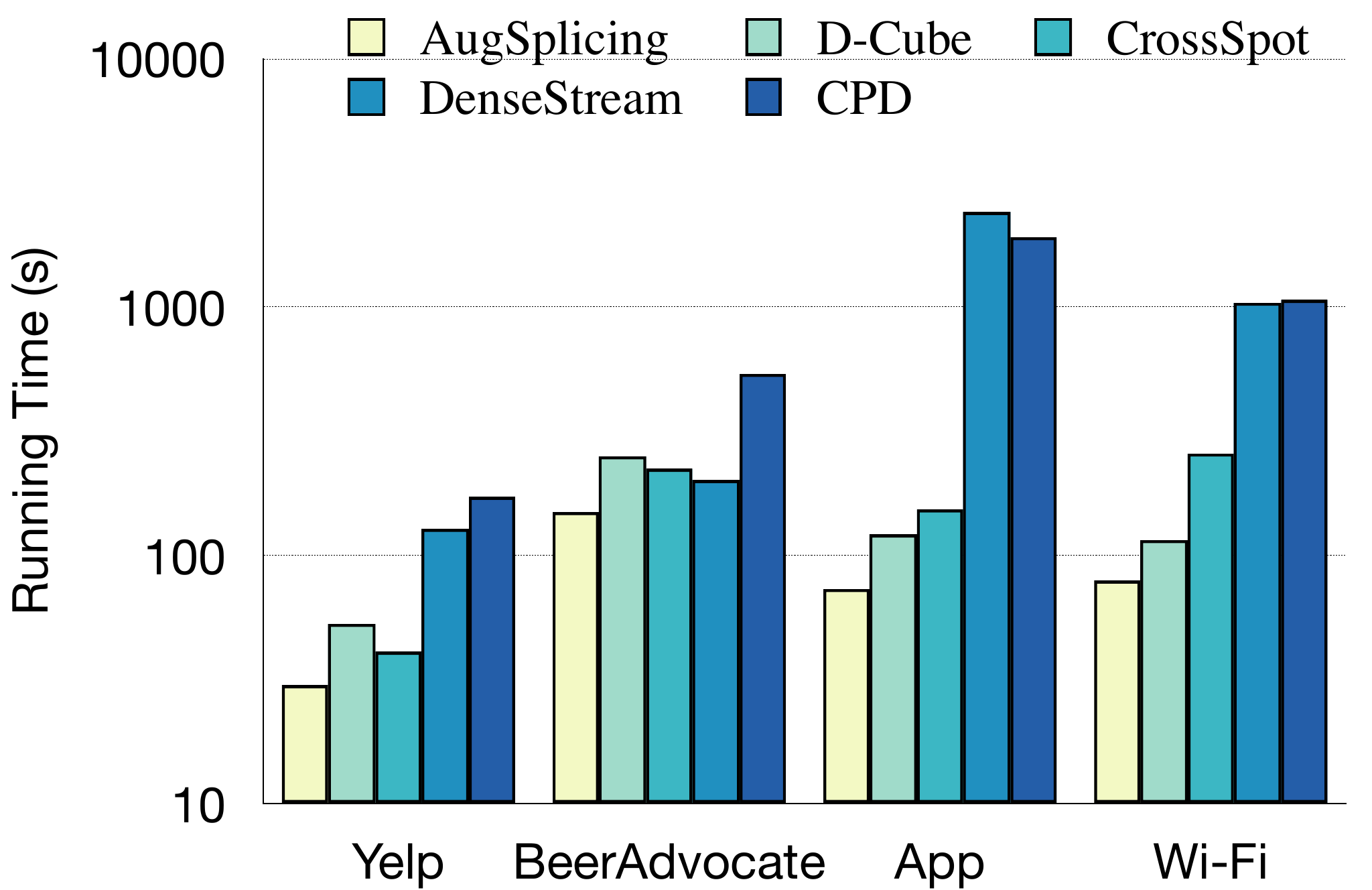}}
    \caption{\method is the fastest.}
    \label{fig:time}
\end{subfigure}
\quad
\begin{subfigure}{0.26\textwidth}
\centering
{\includegraphics[width=\linewidth]{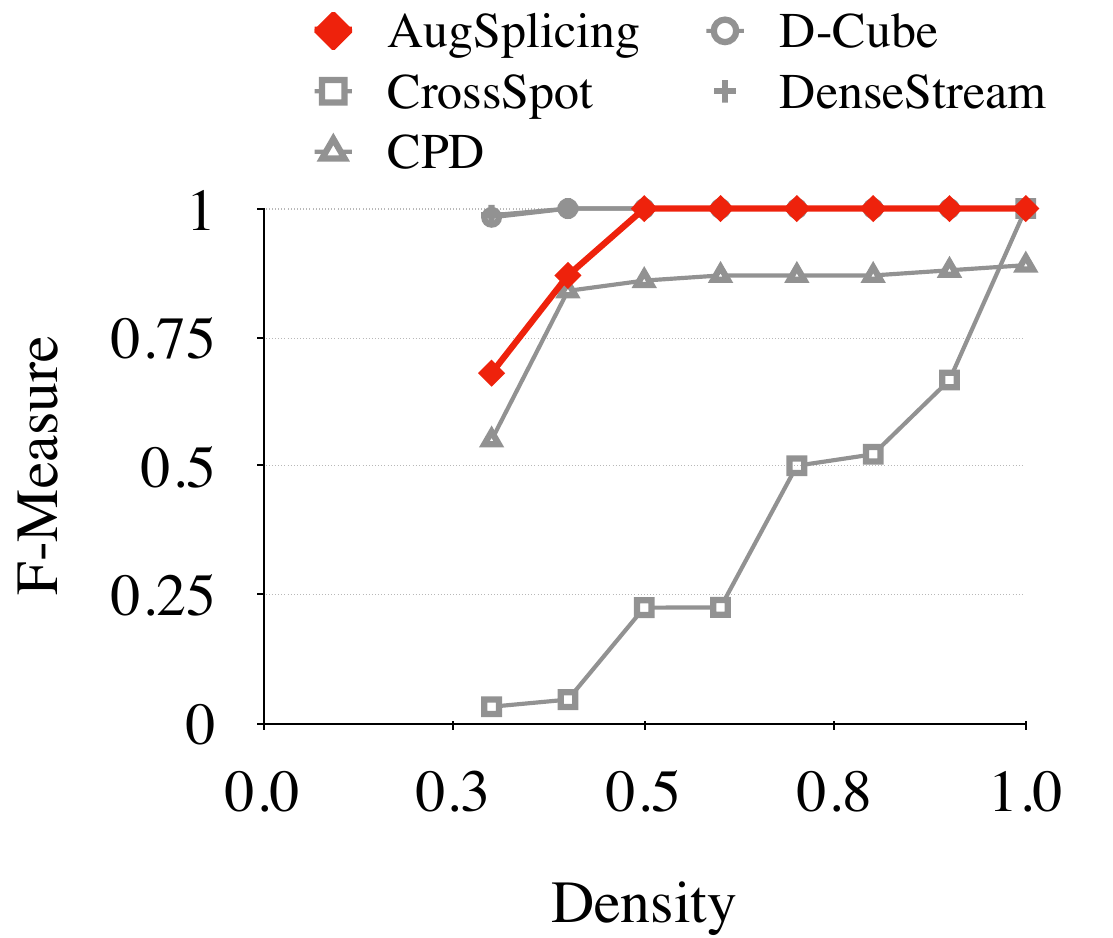}}
\caption{performance on Yelp}
\label{fig:yelp_inject}
\end{subfigure}
\quad
\begin{subfigure}{0.26\textwidth}
  \centering
  {\includegraphics[width=\linewidth]{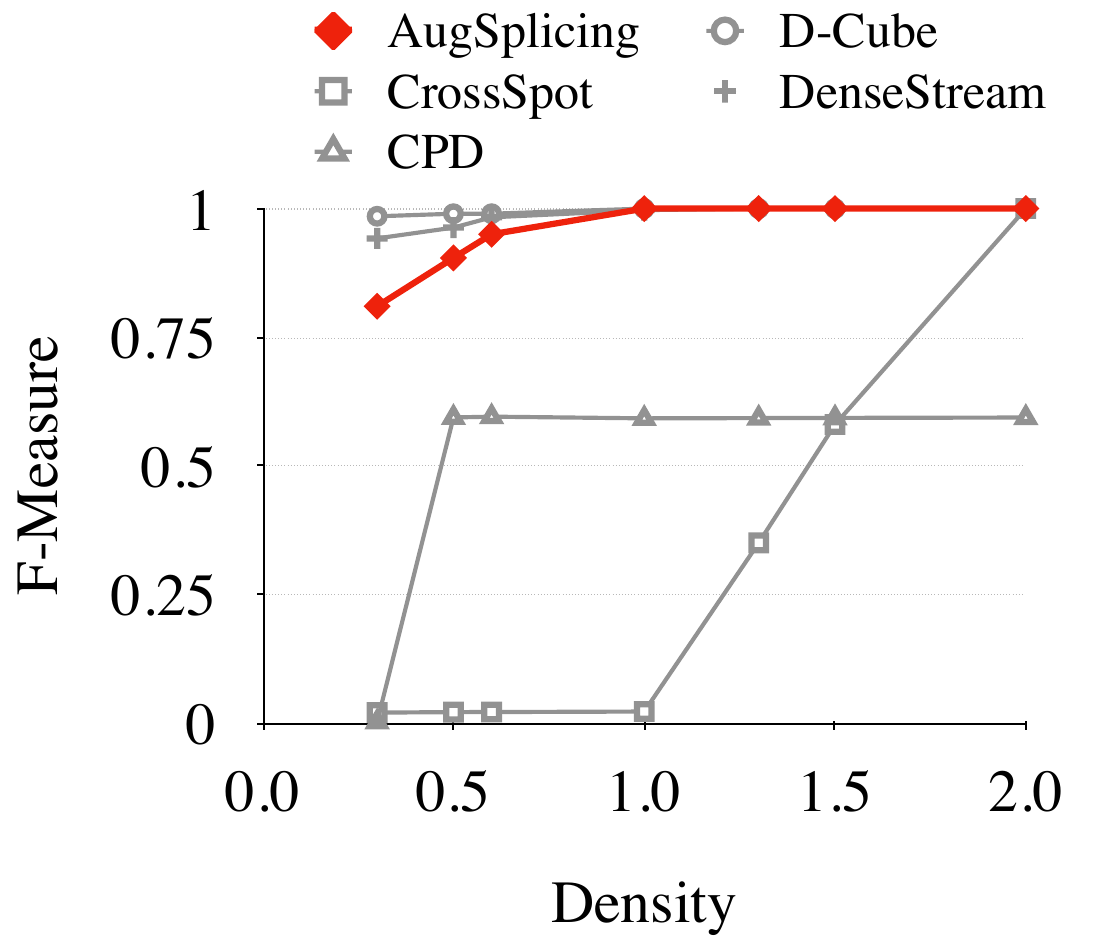}}
  \caption{performance on App}
  \label{fig:app_inject}
 \end{subfigure}
\quad
\caption{\method is fast and accurate. (a) \method\ is $320\times$ faster than baselines. \atn{
	In (b)-(c), our method has accuracy (F-measure) comparable to the state-of-the-art methods: \densestream and \dcube,
    especially when injected fraudulent density is larger than 0.5.
    }
}
\label{fig:result}
\end{figure*}

\begin{figure*}[!htb]
\centering
\begin{subfigure}{0.24\textwidth}
    \centering
    {\includegraphics[width=\linewidth]{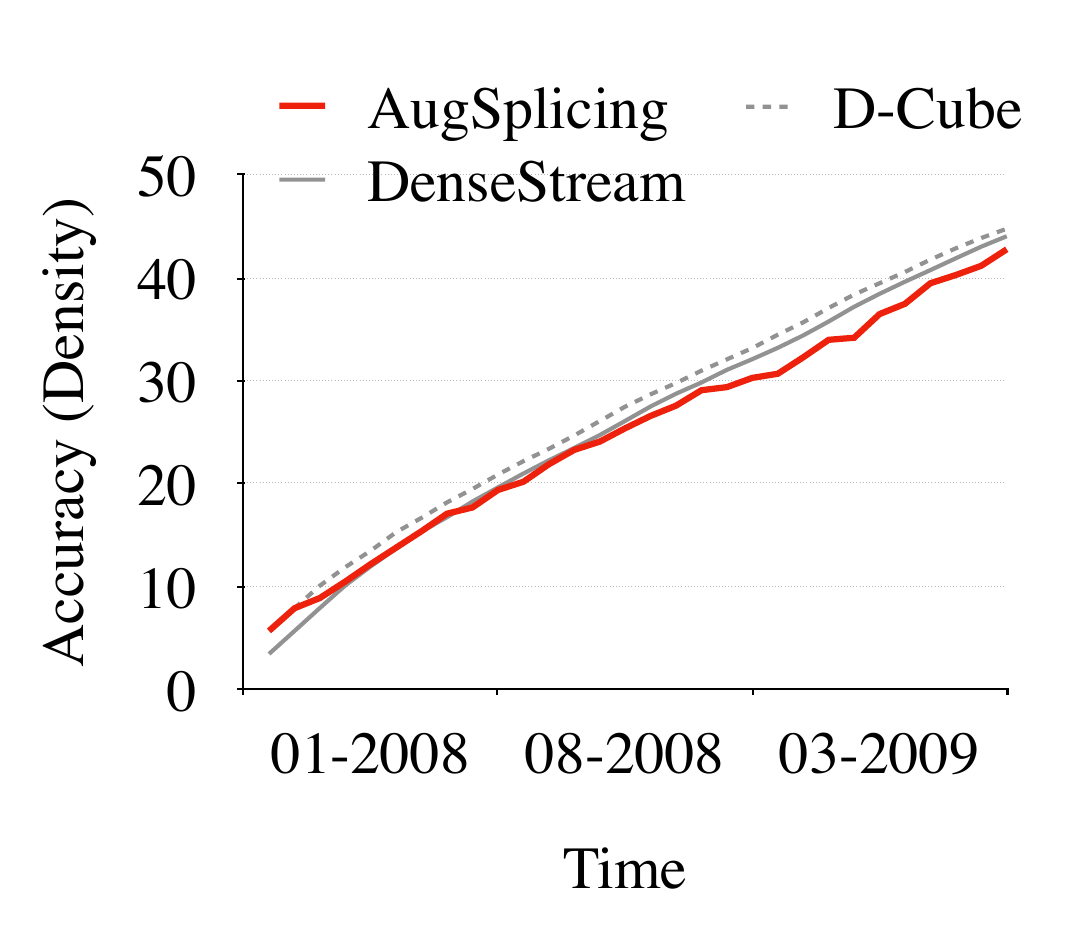}}
    \caption{BeerAdvocate data}
    \label{fig:beer_density}
\end{subfigure}
\begin{subfigure}{0.24\textwidth}
    \centering
    {\includegraphics[width=\linewidth]{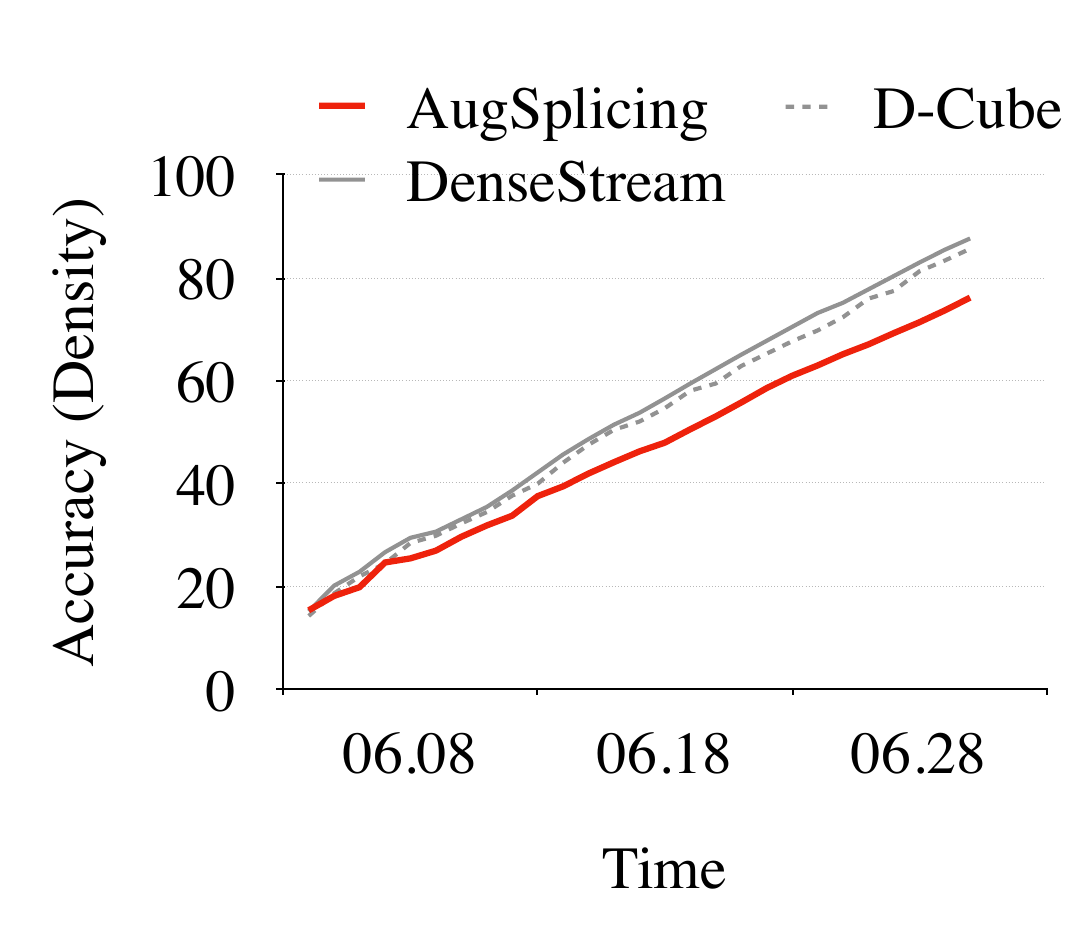}}
    \caption{App data}
    \label{fig:app_density}
\end{subfigure}
\begin{subfigure}{0.24\textwidth}
\includegraphics[width=\linewidth]{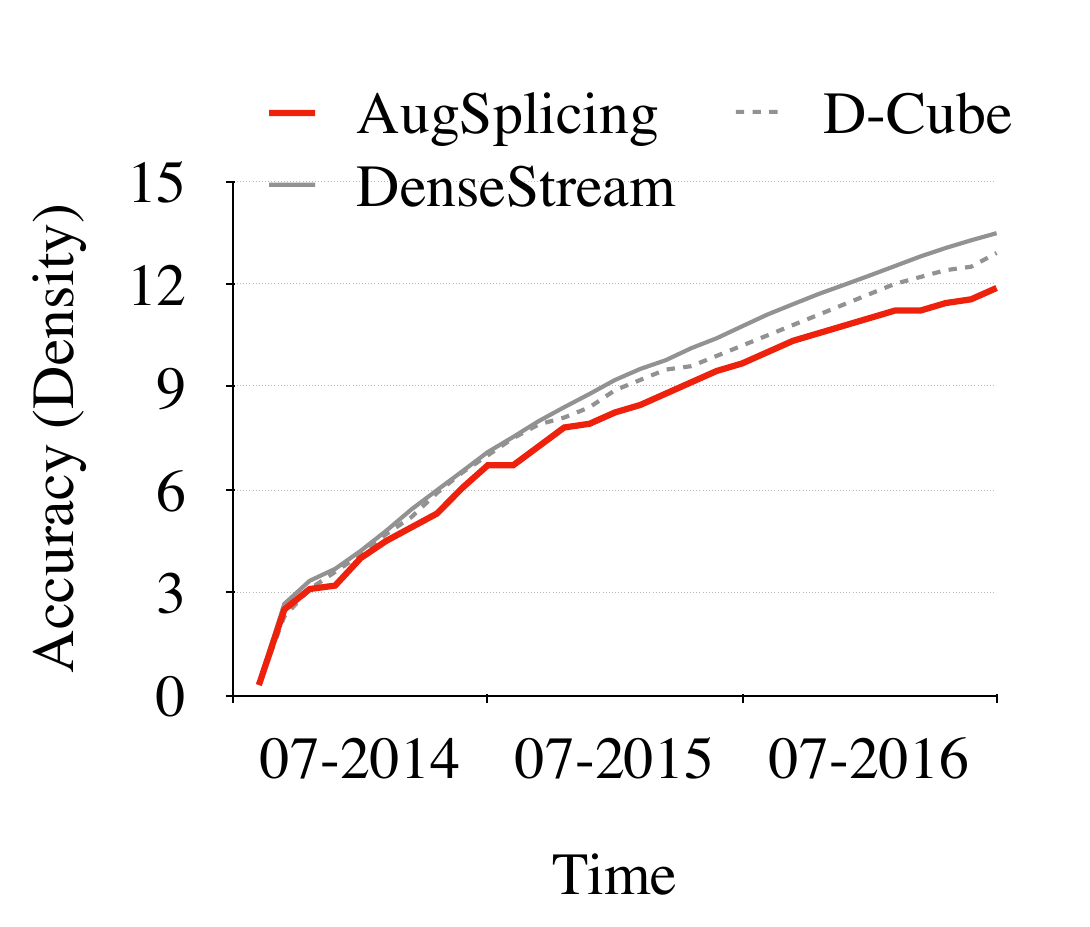}
\caption{Yelp data}
\label{fig:yelp_density}
\end{subfigure}
\begin{subfigure}{0.24\textwidth}
\includegraphics[width=\linewidth]{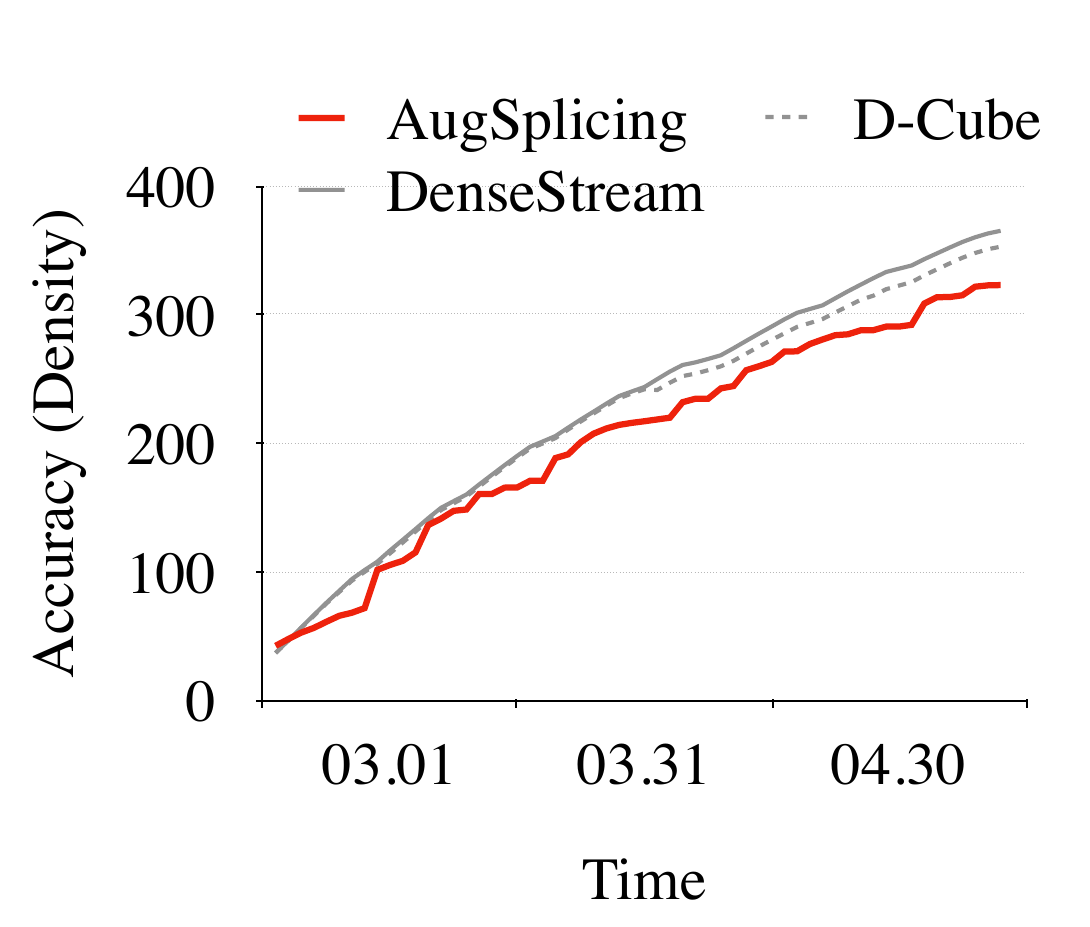}
\caption{Wi-Fi data}
\label{fig:wifi_density}
\end{subfigure}
\caption{\method has comparable accuracy (density) with the state-of-the-art methods. 
}
\label{fig:density}
\end{figure*}

We design the experiments to answer the following questions:

\textbf{Q1. Speed and Accuracy: } How fast and accurate does \algom run compared to the state-of-the-art streaming algorithms and the re-run of batch algorithms on real data? 

\textbf{Q2. Real-World Effectiveness:} 
Which anomalies or lockstep behavior does \method spot in real data?

\textbf{Q3. Scalability:} How does the running time of \algom increase as input tensor grows?

\begin{table}[!t]
\begin{center}

\begin{tabular}[width = \columnwidth]{@{~~}l@{~~}|@{~~}l@{~~}|@{~~}l@{~~}@{~~}l@{~~}}
\toprule
{\bf Name}  & {\bf Volume} & {\bf \# Edges}\\
\midrule
\multicolumn{4}{l}{Rating data (user, item, timestamp)}\\
\midrule
Yelp & 468K $\times$ 73.3K $\times$ 0.73K & 1.34M \\
BeerAdvocate  & 26.5K$\times$50.8K$\times$0.5K & 1.08M \\
\midrule
\multicolumn{4}{l}{mobile devices, app, installing time, uninstalling time}          \\
\midrule
App  & 2.47M$\times$17.9K$\times$30$\times$30  & 5.43M \\
\midrule
\multicolumn{4}{l}{device IP, Wi-Fi AP, connecting time, disconnecting time} \\
\midrule
Wi-Fi  & 119.4K$\times$0.84K$\times$1.46K$\times$1.46K & 6.42M \\
\bottomrule
\end{tabular}
\caption{Data Statistics.}
\label{tab:data}
\end{center}
\end{table}

\paragraph{Datasets:}
Table~\ref{tab:data} lists the real-world data used in our paper. Two rating data are publicly available. 
App data is mobile device-app installation and
uninstallation data under an NDA agreement from a company. Wi-Fi data is device-AP connection and disconnection data from Tsinghua University.

\paragraph{Experimental Setup:}
All experiments are carried out on a $2.3 GHz$  Intel Core i5 CPU with $8 GB$ memory. We compare our method with state-of-the-art streaming dense block detection method, \densestream, and the re-run of batch methods, \dcube, \cross, and CP Decomposition (\cpd). \dcube is implemented in Java to detect dense blocks \atn{in tensor $\TX(t,s)$}. Specifically, we use ``arithmetic average mass'' as the metric of \dcube. We use a variant of \cross which maximizes the same metric and use the \cpd result for seed selection similar to \cite{shin2016mzoom}.
We set time stride $s$ to $30$ in a day for Yelp, $15$ in a day for BeerAdvocate, $1$ in a day for App and Wi-Fi, as
different time granularity. $k$ is set to $10$ and $l$ to $5$ for all datasets.

\begin{figure*}[t]
\centering
\begin{subfigure}[t]{0.24\textwidth}
\centering
{\includegraphics[width=\linewidth]{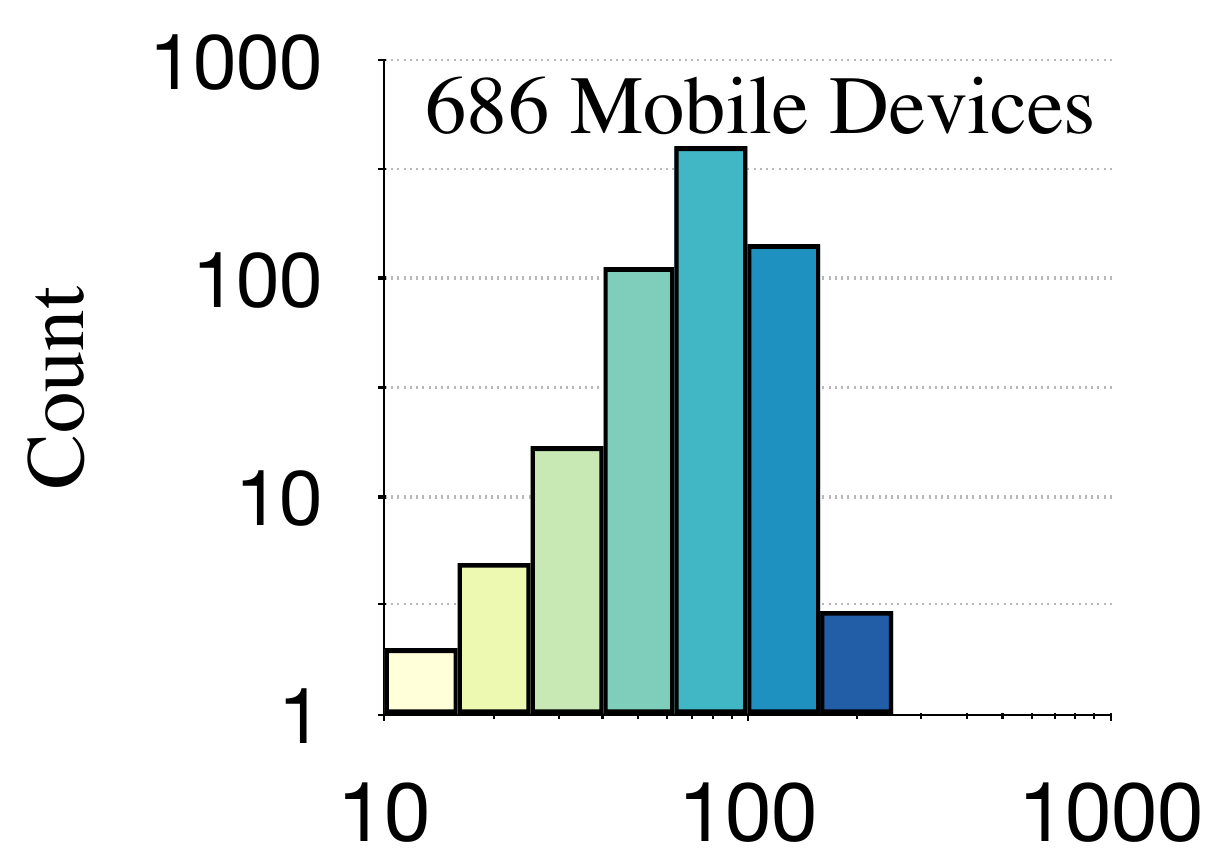}}
\caption{Degree of Devices}
\label{fig:appcase1}
\end{subfigure}
\begin{subfigure}[t]{0.24\textwidth}
\centering
{\includegraphics[width=\linewidth]{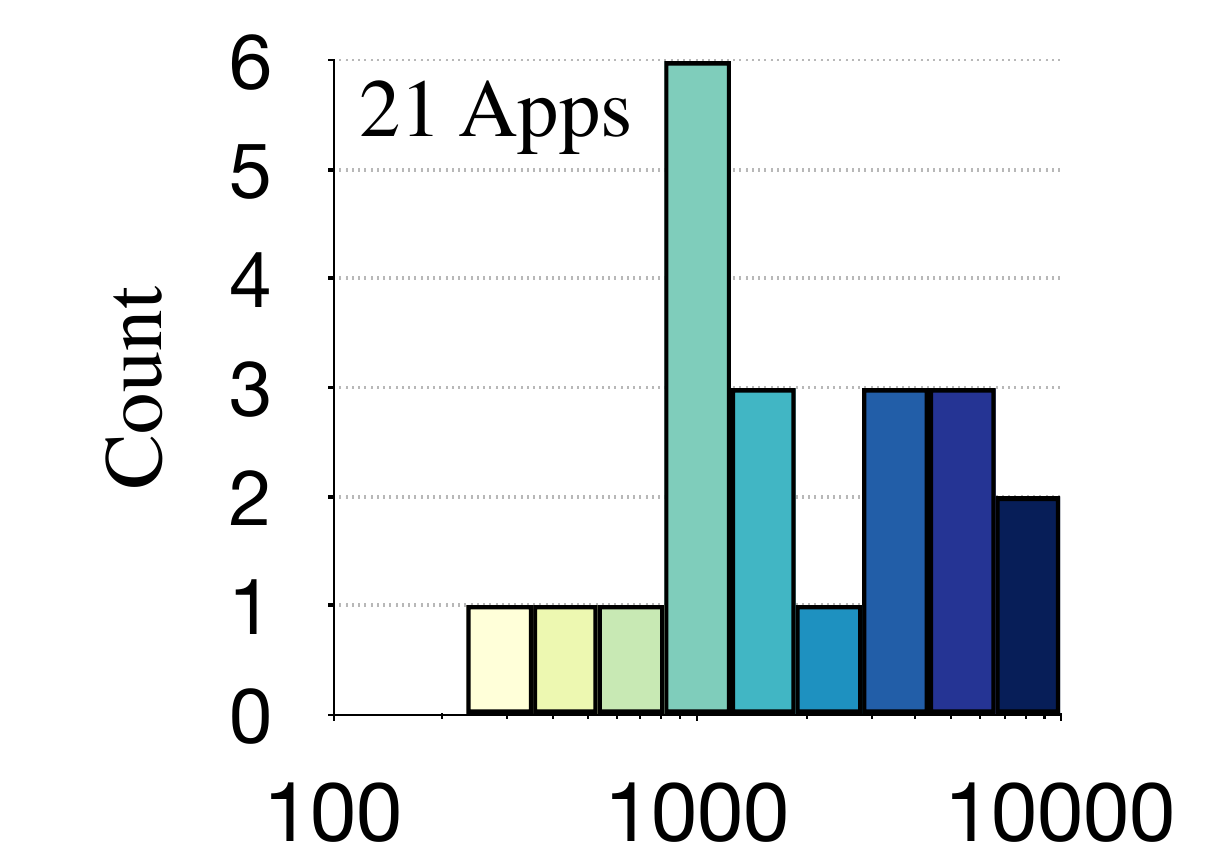}}
\caption{Degree of Apps}
\label{fig:appcase2}
\end{subfigure}
\begin{subfigure}[t]{0.24\textwidth}
\centering
{\includegraphics[width=\linewidth]{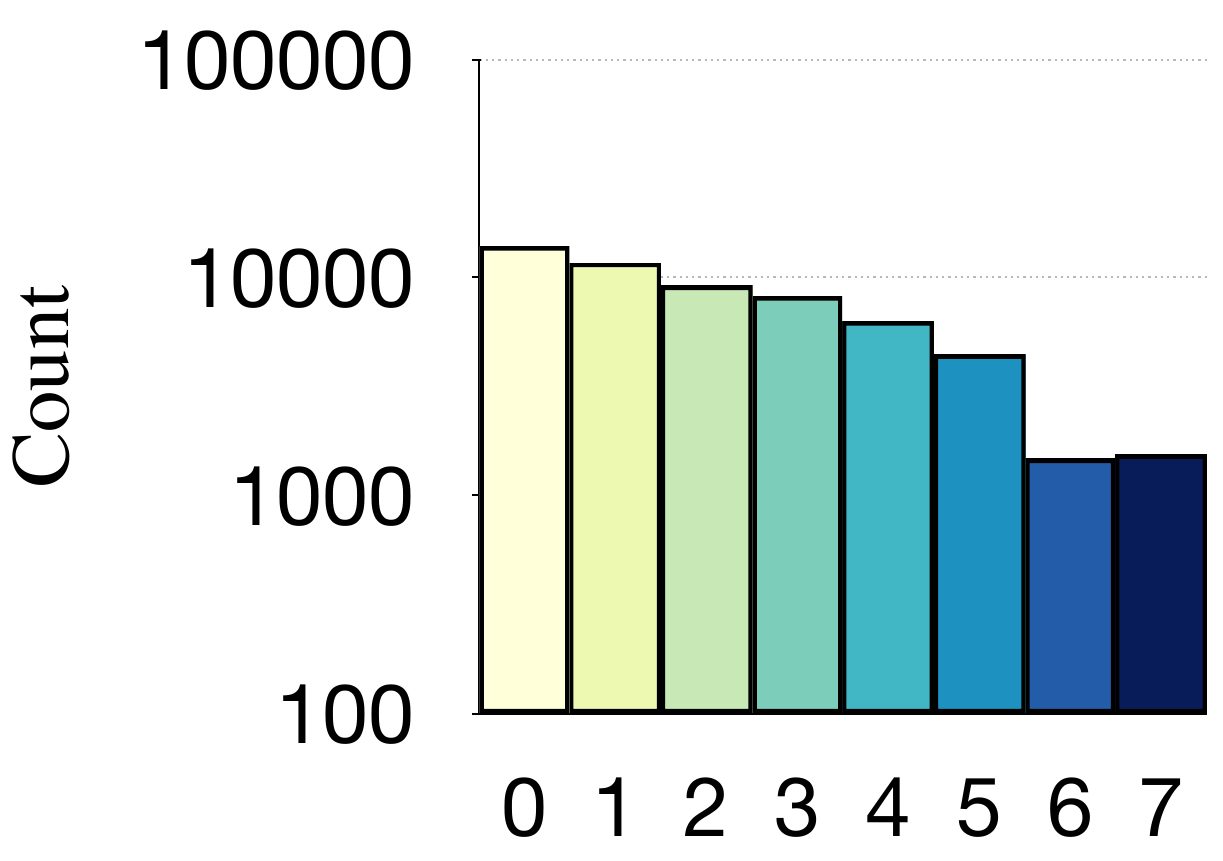}}
\caption{Installing Times}
\label{fig:appcase3}
\end{subfigure}
\begin{subfigure}[t]{0.21\textwidth}
\centering
{\includegraphics[width=\linewidth]{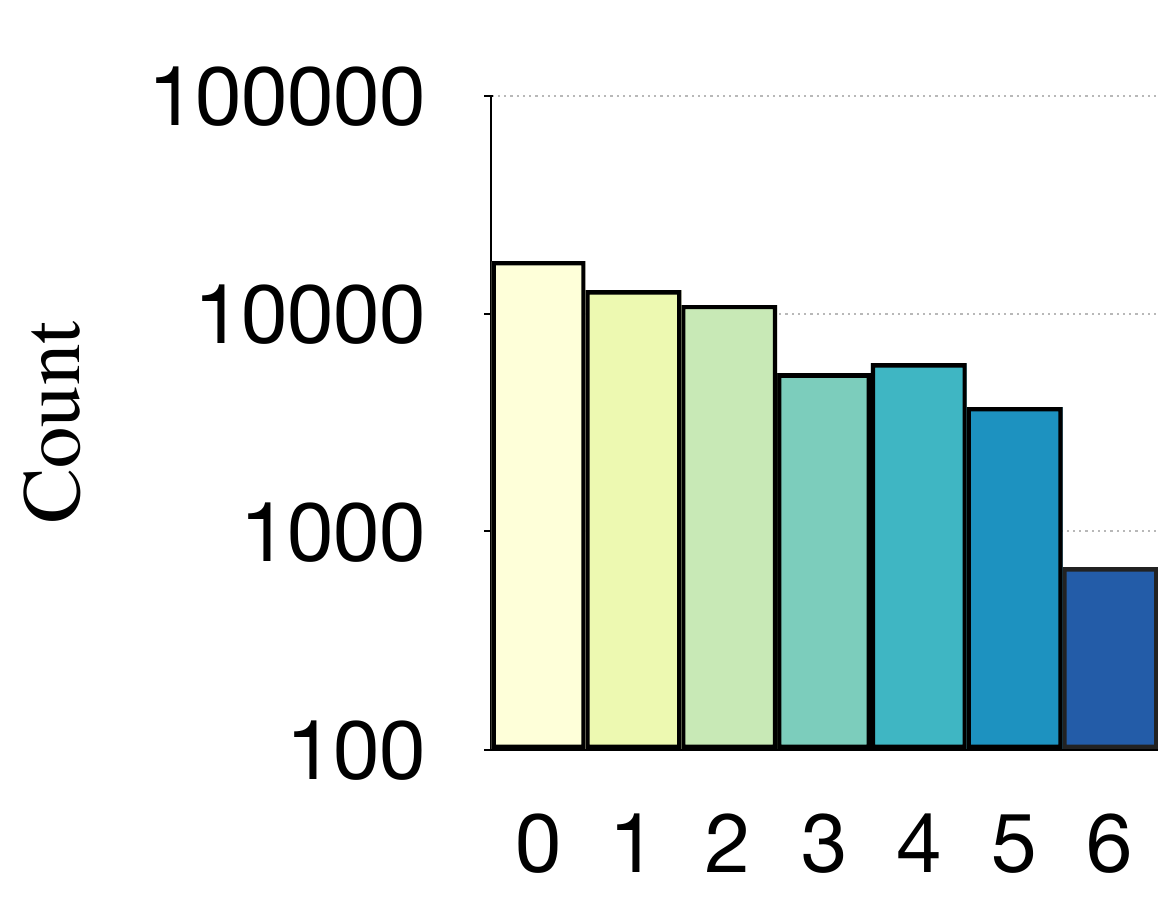}}
\caption{Staying Times}
\label{fig:appcase4}
\end{subfigure}
\caption{(a) and (b) show \method detects a real-world suspicious block that is explainable: $686$ mobile devices repeatedly installed and uninstalled $21$ apps $5.66\times10^4$ times in total, which is very unusual for a group of devices and apps. (c) and (d) show an 8-day installing time period and all that suspicious apps were uninstalled within one week, and most of them stayed only up to 3 days on a suspicious device.}
\end{figure*}

\begin{figure*}[t]
\begin{subfigure}[t]{0.24\textwidth}
\centering
{\includegraphics[width=\linewidth]{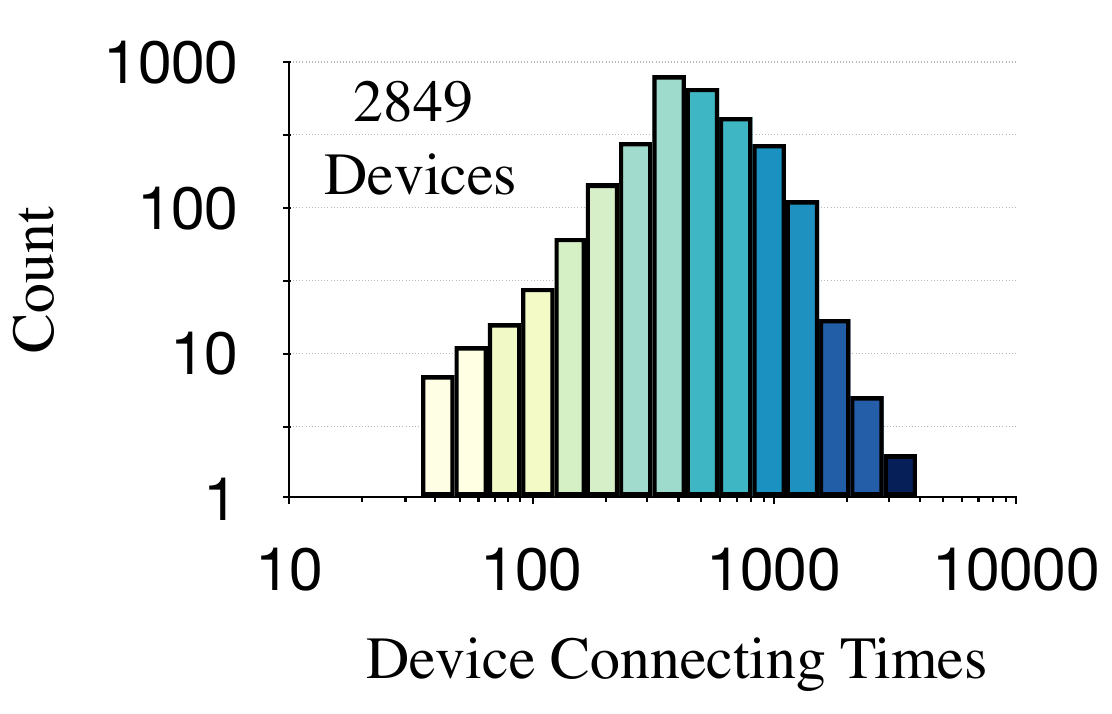}}
\caption{Degree of Devices}
\label{fig:wificase1}
\end{subfigure}
\begin{subfigure}[t]{0.24\textwidth}
\centering
{\includegraphics[width=\linewidth]{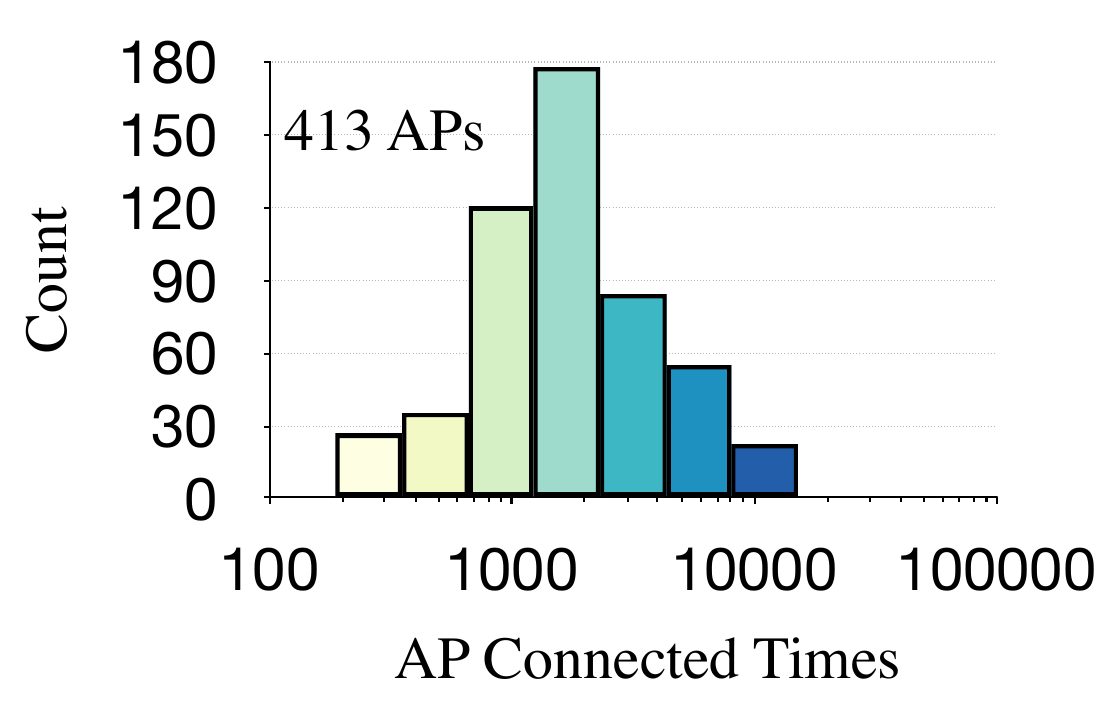}}
\caption{Degree of Wi-Fi APs}
\label{fig:wificase2}
\end{subfigure}
\begin{subfigure}[t]{0.24\textwidth}
\centering
{\includegraphics[width=\linewidth]{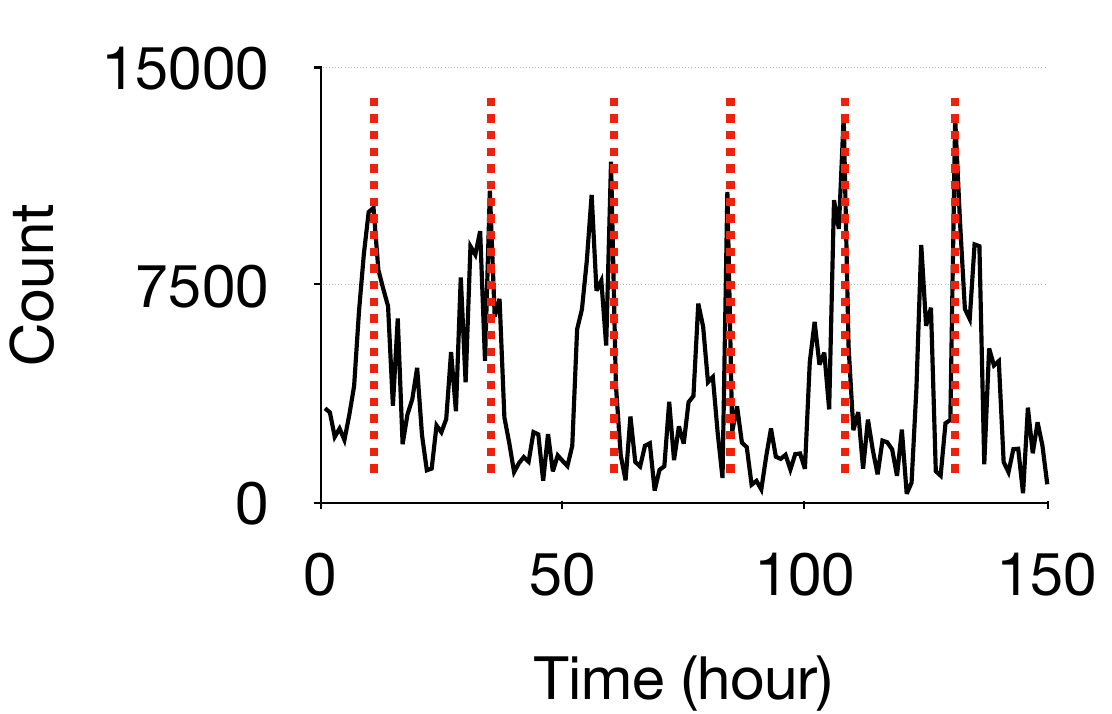}}
\caption{Connecting Time}
\label{fig:wificase3}
\end{subfigure}
\begin{subfigure}[t]{0.24\textwidth}
\centering
{\includegraphics[width=\linewidth]{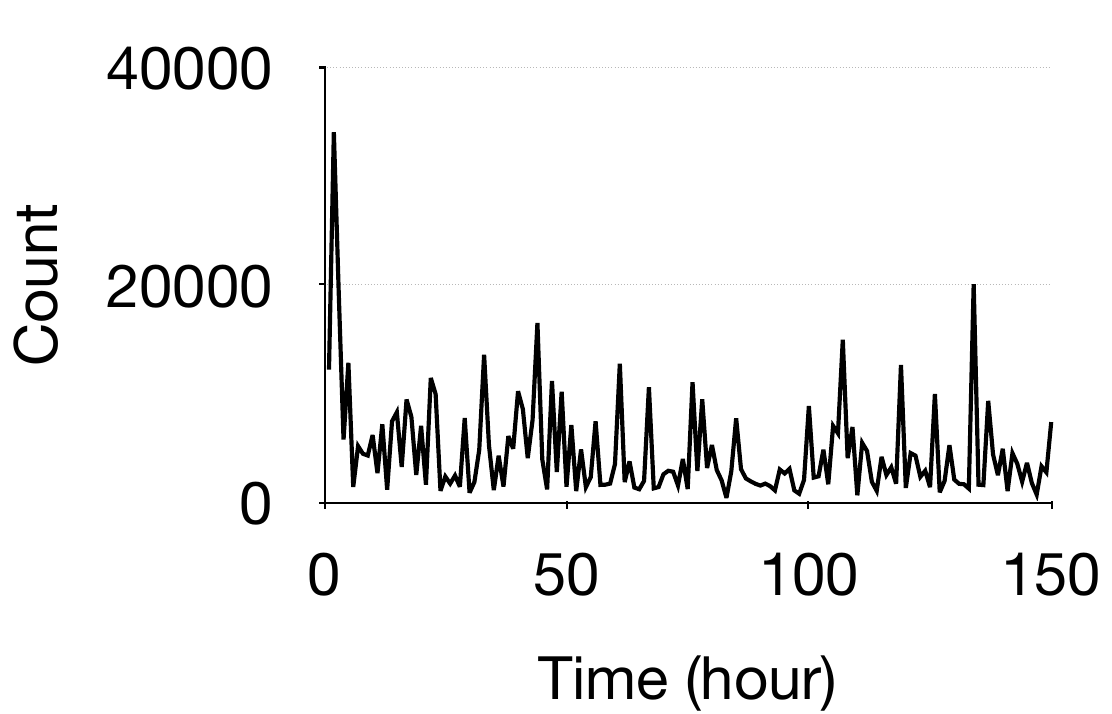}}
\caption{Staying Time}
\label{fig:wificase4}
\end{subfigure}
\caption{ \method finds a big community of students having synchronized busy schedule in Wi-Fi data. Wi-Fi connections reached the peak at around $10$ AM every day (red dash
    lines) in (c). Most of the connections stayed around $1$ to $2$ hours as shown in
    (d).}
\label{fig:wifi_distribution}
\end{figure*}

\subsection{Speed and Accuracy}
For detecting dense blocks in a streaming setting, our method only deals with augmenting tensor with \emph{stride} size at each time step, then it combines detected blocks of the incremental tensor with the previous results to detect dense blocks until now. In contrast, the batch algorithms are re-run for the holistic tensor from scratch at each time step
to detect dense blocks of augmenting tensors. \densestream needs to maintain a dense subtensor when a new entry comes, which is very time-consuming. 
We measure the wall-clock time taken by each method and the results are as shown in Figure~\ref{fig:time}. As we can see, 
\method is the fastest. It is $320\times$ faster than \densestream, $1.8\times$ faster than \dcube, $3.2\times$ faster than \cross and $13\times$ faster than \cpd on Wi-Fi dataset.

To demonstrate the accuracy of \algom, we track the density of the densest
block found by \method\ and other methods while the tensor augments at each time step \zhangatn{as~\cite{shin2017densealert} does} and the result is shown in
Figure~\ref{fig:beer_density}-\ref{fig:wifi_density}. We can see that the densest block has close density to that found by \densestream and the re-run of \dcube for long time steps, though accumulated error.

We now explain why \method achieves comparable high accuracy. Due to the skewness of real graphs, densities of top
dense blocks can be very skewed, which reduces the probability of the top $k$
dense blocks of $\TX(t+s)$ having overlapped modes with top $(k+l)$ or lower dense blocks in $\TX(t)$ and $\TX(t, s)$. 
Due to the principle of time locality, tuples of dense blocks will be close in mode $time$. Thus \method can detect top $k$ dense blocks with comparable high density by sufficient splices.

\textbf{Detection in injected attacks:}
For Yelp data, we injected $100$ fraudulent users and items in a week with the volume density ranging from $1$ to $0.1$. \atn{For app data, an app's rank is decided by its downloads, which improves by $1$ if the app is installed and remains more than the required amount of days by a mobile device in a real scenario.} Then we injected $500$ apps and $5000$ mobile devices, with the installing time uniformly distributed in $3$ days. The uninstalling time was the installing time plus a week with the volume density ranging from $2.0$ to $0.1$. \zhangatn{Intuitively, the smaller the density of injected blocks, the harder it is to detect, and the block with a density of $0.1$ is quite difficult to detect.} Figures~\ref{fig:yelp_inject}-\ref{fig:app_inject} show that F-measure of \method increases rapidly as density increases from $0.3$ to $0.5$ and remains higher than 90$\%$ when the density reaches to $0.5$, \zhangatn{presenting our ability in detecting fraudulent mobile devices and apps.}

\subsection{Effectiveness}
\textbf{Results on App data \zhangatn{with ground-truth labels}:}
\zhangatn{In this section, we verify that \method accurately detects a dense block of fraudulent accounts in App data, as verified by clear signs of fraud exhibited by a majority of detected mobile devices and apps.}
We collected the devices detected by all methods and manually labeled by the company who owns App data, 
based on empirical knowledge on features of devices: e.g. locations,
the number of installed apps, 
and the number of apps opened for the first time in a day, etc. 
For example, devices are identified as fraudulent if 
they appear in dozen cities in a day.

Figure~\ref{fig:realaccurate} shows 
both accuracy (F-measure) and speed (elapsed running time) for all comparison methods. We can see that the \method runs $320\times$ faster than the state-of-the-art streaming algorithm, \densestream, keeping comparable accuracy. Compared to the fast re-run \dcube, \method achieves $1.8\times$ faster, and much higher accuracy. Figures~\ref{fig:appcase1}-\ref{fig:appcase4} present the detailed information of the densest block detected by \method. We draw ``degree''
distributions for detected $686$ devices and $21$ apps in Figure~\ref{fig:appcase1} and \ref{fig:appcase2}. Note that the ``degree'' of a mobile device is \# of apps
installed by the mobile. Similarly ``degree'' of an app is the number of devices installing the app.
As a result, $365$ devices from $686$ detected devices have been identified as fraudsters by the company, which is a very high concentration in a fraud detection task, considering a small fraction of fraudulent devices over the whole devices. 
Actually, devices not identified as fraudsters are very suspicious by analyzing their behavior: $686$ mobile devices repeatedly installed and uninstalled $21$ apps $5.66\times10^4$ times, with up to $7100$ times for one app. Furthermore, all the installations were concentrated in a short time period (i.e.~$8$ days) and uninstalled within one week afterward (see Figures~\ref{fig:appcase3}-\ref{fig:appcase4}). It is very likely that these mobile devices boost these apps'
ranks in an app store by installations and uninstallations in lockstep.

\textbf{Results on Wi-Fi data:}
We discover synchronized patterns that may interest student administrators.
Figure~\ref{fig:wifi_distribution} shows the densest block detected by \method in Wi-Fi data.  
Figure~\ref{fig:wifi_distribution}a and
\ref{fig:wifi_distribution}b show $2849$ devices and $413$ Wi-Fi APs which had $8.03\times 10^{5}$ connections/disconnections in total, indicating a busy schedule for many students.
As shown in Figure~\ref{fig:wifi_distribution}c, behavior of this group of students was periodic and synchronized. 
Wi-Fi connections concentrated from $8$ AM - $5$ PM every day and reached a peak at around $10$ AM (red dotted line). 
That may be because students' first and second classes begin around $8$ AM and $10$ AM respectively.
Moreover, Figure~\ref{fig:wifi_distribution}d shows that most of the connections stayed around $1$-$2$ hours, which is the usual length of one class. 

 \subsection{Scalability}
 
Figure~\ref{fig:app_scala} and \ref{fig:wifi_scala} show that the running time of \algom scales linearly with the size of non-zero entries up to the current time step, consistent with the time complexity analysis result in Section 4.3, while the running time of the re-run of \dcube is quadratic in streaming tensors.
\begin{figure}[!t]
\centering
\begin{subfigure}{0.23\textwidth}
\includegraphics[width=\linewidth]{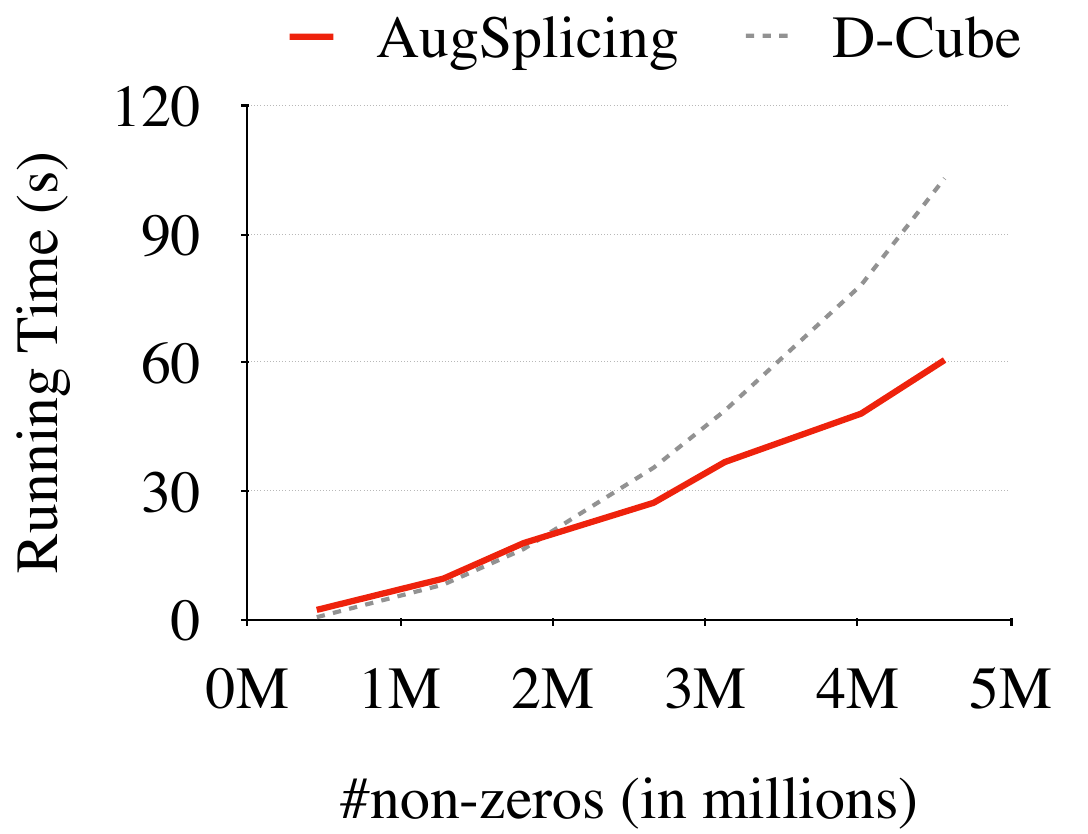}
\caption{App data}
\label{fig:app_scala}
\end{subfigure}
\begin{subfigure}{0.23\textwidth}
\includegraphics[width=\linewidth]{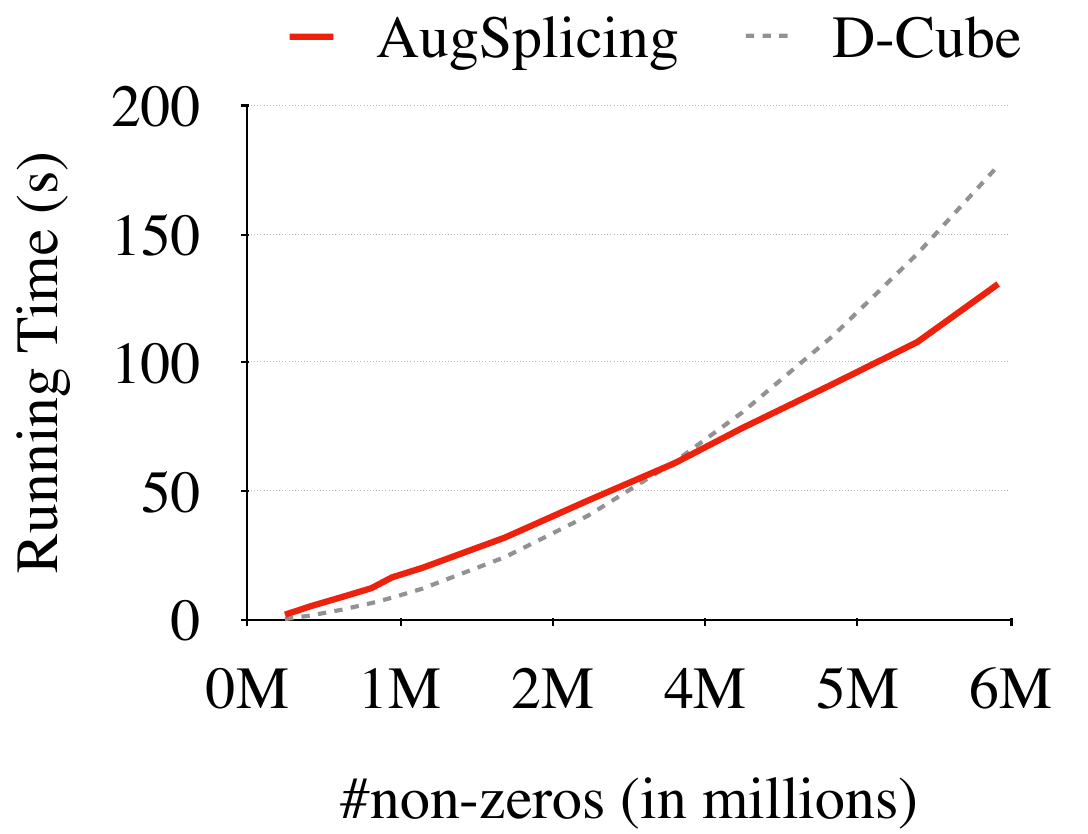}
\caption{Wi-Fi data}
\label{fig:wifi_scala}
\end{subfigure}
\caption{\method has near-linear running time. 
}
\label{fig:scala}
\end{figure}


\section{Conclusions}
\label{sec:cons}
In this paper, we model a stream of tuples as a streaming tensor, and propose a streaming algorithm, \method, to spot the most synchronized behavior which indicates anomalies or interesting patterns efficiently and effectively. Our main contributions are as follows: 
\begin{enumerate}
    \item \textbf{Fast and streaming algorithm:} Our approach can effectively capture synchronized activity in streams, up to $320\times$ faster than the best streaming method (Figure~\ref{fig:realaccurate}).
    \item \textbf{Robustness:} Our method is robust with theory-guided incremental splices for dense block detection.
    \item \textbf{Effectiveness:} Our method is able to detect anomalies and discover interesting patterns accurately in real-world data (Figures~\ref{fig:appcase1}-\ref{fig:appcase4} and Figures~\ref{fig:wificase1}-\ref{fig:wificase4}).
\end{enumerate}

\clearpage

\section*{Ethical Impact}
\label{sec:impact}
We contribute a fast incremental algorithm to detect dense blocks formed by synchronized behavior in a stream of time-stamped tuples. Our work has wide applications on anomaly detection tasks, e.g.~app store fraud detection (e.g.~suspicious mobile devices boosting target apps' ranks in recommendation list in an app store), rating fraud detection in review cites and etc. In addition, our work can be applied to discover interesting patterns or communities in real data (e.g.~revealing a group of students having the same classes of interest). Our approach can scale to very large data, update the estimations when data changes over time efficiently, and incorporate all the information effectively. Our work is even more applicable to online data mining tasks, especially when large-sized data arrives at a high rate. While most experiments are cyber-security related, one experiment detects student communities from Wi-Fi data. From a societal impact, potential misuse against privacy has to be taken care of.

\section*{Acknowledgments}
\label{sec:ack}
This paper is partially supported by the National Science Foundation of China under Grant No.U1911401, 61772498, 61872206, 91746301. This paper is also supported by the Strategic Priority Research Program of the Chinese Academy of Sciences, Grant No. XDA19020400, National Key R\&D Program of China (2020AAA0103502) and 2020 Tencent Wechat Rhino-Bird Focused Research Program.

\bibliography{paper}

\end{document}